\pdfoutput=1

\documentclass[11pt]{article}
\usepackage[utf8]{inputenc}
\usepackage[english]{babel}
\usepackage{blindtext}
\usepackage{tcolorbox}
\usepackage{graphicx}
\usepackage{import}
\usepackage{amsmath}
\usepackage{amsthm}
\usepackage{amssymb}
\usepackage{bbm}
\usepackage{esint}
\usepackage{cancel}
\usepackage{stmaryrd}
\usepackage{asymptote}
\usepackage{todonotes}
\usepackage[colorlinks=true, citecolor=blue]{hyperref}
\usepackage[capitalize]{cleveref}

\newtheorem{theorem}{Theorem}[section]

\newtheorem{lemma}[theorem]{Lemma}

\newtheorem{example}[theorem]{Example}

\newtheorem{definition}[theorem]{Definition}
\newtheorem{remark}[theorem]{Remark}

\begin{document}
\title{The trace reconstruction problem for spider graphs}

\author{Alec Sun\thanks{Carnegie Mellon University, %Pittsburgh, PA 15213, USA,
\texttt{alecsun@andrew.cmu.edu}} \and William Yue\thanks{Massachusetts Institute of Technology, %Cambridge, MA 02139, USA, 
\texttt{willyue@mit.edu}}}

\maketitle

\begin{abstract}
We study the trace reconstruction problem for spider graphs. Let $n$ be the number of nodes of a spider and $d$ be the length of each leg, and suppose that we are given independent traces of the spider from a deletion channel in which each non-root node is deleted with probability $q$. This is a natural generalization of the string trace reconstruction problem in theoretical computer science, which corresponds to the special case where the spider has one leg. In the regime where $d\ge \log_{1/q}(n)$, the problem can be reduced to the vanilla string trace reconstruction problem. We thus study the more interesting regime $d\le \log_{1/q}(n)$, in which entire legs of the spider are deleted with non-negligible probability. We describe an algorithm that reconstructs spiders with high probability using $\exp\left(\mathcal{O}\left(\frac{(nq^d)^{1/3}}{d^{1/3}}(\log n)^{2/3}\right)\right)$ traces. Our algorithm works for all deletion probabilities $q\in(0,1)$.

\textbf{Keywords:} Trace reconstruction, Graph algorithms,  Littlewood polynomials
\end{abstract}

\section{Introduction}
The \emph{string trace reconstruction problem}, first introduced in 1997 by Levenshtein \cite{levenshtein1997reconstruction}, is concerned with reconstructing an unknown \emph{seed} string using only noisy samples of the data. The unknown seed string is passed into some noisy channel multiple times, and the resulting error-prone copies are referred to as \emph{traces}. The goal is to use multiple traces to reconstruct the original seed string with high probability. Levenshtein solved the trace reconstruction problem for a \emph{substitution channel}, where each symbol of the seed string is mutated independently with constant probability. In 2004, Batu, Kannan, Khanna, and McGregor \cite{batu2004reconstructing} analyzed the problem for a \emph{deletion channel}, where symbols of the seed string are each deleted independently with constant probability. The string trace reconstruction problem has applications to computational biology, specifically in the new rapidly-evolving fields of DNA data storage and personalized immunogenics. For example, one might want to reconstruct the correct sequence of nucleotides of a DNA sequence from several traces, each of which has many deletion mutations.

It is critical to minimize the number of traces required to reconstruct the seed string with high probability. For example, in the application of DNA data storage, reducing the number of traces results in lower sequencing cost and time \cite{bhardwaj2020trace}. However, despite a wealth of recent work and attention on the deletion channel string trace reconstruction problem, for example \cite{chase2021new,de2017optimal,hartung2018trace,holden2020lower,holden2020subpolynomial,holenstein2008trace,mcgregor2014trace,nazarov2017trace,viswanathan2008improved}, the current best upper and lower bounds for the number of traces necessary to reconstruct the seed string with high probability remain at $\exp(\mathcal{O}(n^{1/5}))$ \cite{chase2020new} and $\tilde\Omega(n^{3/2})$  \cite{chase2021new,holden2020lower}, respectively, where $n$ is the length of the seed string. We remark that a lower bound of $\exp(\mathcal{O}(n^{1/3}))$ traces was shown for mean-based algorithms, which are algorithms that only use the empirical means of individual bits in the traces for reconstruction \cite{de2017optimal,nazarov2017trace}.

The exponential gap between upper and lower bounds for string trace reconstruction motivates studying variants of the problem for which one may be able to close the gap. Many variants have been recently proposed and studied, for example \cite{andoni2012global,chen2022near,cheraghchi2020coded,davies2021approximate,krishnamurthy2021trace,narayanan2021circular}. We focus on a variant known as the \emph{tree trace reconstruction problem} introduced by Davies, R\'acz, and Rashtchian \cite{davies2019reconstructing}. This is a generalization of the vanilla string trace reconstruction problem where the goal is to learn a node-labeled tree, rather than a single string, using traces from a suitably-defined deletion channel. The tree trace reconstruction problem may be directly applicable as well, as research on DNA nanotechnology has demonstrated that DNA molecule structures can be assembled into trees. Recent research has also shown how to distinguish different molecular topologies, such as spiders with three arms from line DNA, using nanopores \cite{karau2018capture}. 

Davies et al. \cite{davies2019reconstructing} studied the tree trace reconstruction problem for two special classes of trees: complete $k$-ary trees and spiders. This paper extends their work on spiders. An \emph{$(n,d)$-spider} consists of a single unlabeled root node with paths of $d$ labeled nodes attached to it. In total, there are $n$ labeled nodes. Consider a deletion channel, formally defined in \cref{deletion channel}, in which every node is independently deleted with probability $q$.

 When $d\ge \log_{1/q}(n)$, solving the spider trace reconstruction problem directly reduces to the string trace reconstruction problem \cite[Proposition 24]{davies2019reconstructing}. This is because in this regime, the legs of the spider are long enough for all of the legs to survive the deletion channel with high probability, so each leg can be considered independently as its own string trace reconstruction problem. Therefore, we assume that $d\le \log_{1/q}(n)$. In this more interesting regime, entire legs are deleted with non-negligible probability. Hence, if one looks at a single trace, it is unclear which of the legs in the seed spider the legs in the trace come from.

Davies et al. \cite{davies2019reconstructing} proved that for deletion probabilities $q<0.7$, there is some constant $C>0$ that depends only on $q$ such that $\exp(C\cdot d(nq^d)^{1/3})$ traces suffice to reconstruct an $(n,d)$-spider with probability $1-\mathcal{O}(1/n)$ (we refer to this as \emph{with high probability}). In this paper, we match this upper bound, up to polylogarithmic factors, but for the full range of deletion probabilities $q\in (0,1)$. Furthermore, while Davies et al. \cite{davies2019reconstructing} used a single variable generating function alongside harmonic analysis, we consider a bivariate generating function, which results in considerably simpler analysis. We use a best-match algorithm coupled with some results about bivariate Littlewood polynomials. We remark that Littlewood polynomials have also been used to analyze a different variant of trace reconstruction known as the \emph{matrix reconstruction problem} \cite{krishnamurthy2021trace}.

Our main result is the following theorem:

\begin{theorem}\label{main result}
Assume that $d\le \log_{1/q}(n)$. For any fixed deletion probability $q<1$, there exists some constant $C>0$ that depends only on $q$ such that
\[\exp\left(C\cdot \frac{(nq^d)^{1/3}}{d^{1/3}}(\log n)^{2/3}\right)\]
traces suffice to reconstruct an $(n,d)$-spider with high probability. 
\end{theorem}

Note that the upper bound in \cref{main result} matches the upper bound $\exp(C\cdot d(nq^d)^{1/3})$ in \cite{davies2019reconstructing} up to polylogarithmic factors and works for all deletion probabilities $q\in (0, 1)$, not just $q<0.7$. Furthermore, \cref{main result} strictly improves upon the upper bound $\exp(C\cdot d(nq^d)^{1/3})$ for all $q\in (0,1)$ when $d = \omega(\sqrt{\log n})$.

\subsection{Acknowledgements} The authors would like to thank Shyam Narayanan for suggesting the problem.

\section{Preliminaries}

\subsection{Rooted spiders}

In this section, we define the objects to be reconstructed: rooted binary-labeled spiders $X$, as well as an indexing system for their nodes.

\begin{definition}\label{rooted spiders}
    Let $n$ and $d$ be positive integers, and for convenience assume that $d\mid n$. An \emph{$(n,d)$-spider} $X$ consists of a single unlabeled root node with $\frac{n}{d}$ paths of $d$ nodes with binary labels from $\{0,1\}$ emanating from it, so there are $n$ labeled nodes in total. We refer to these paths as the \emph{legs} of the spider.
\end{definition}

\begin{figure}
    \centering
    \includegraphics[width=7cm]{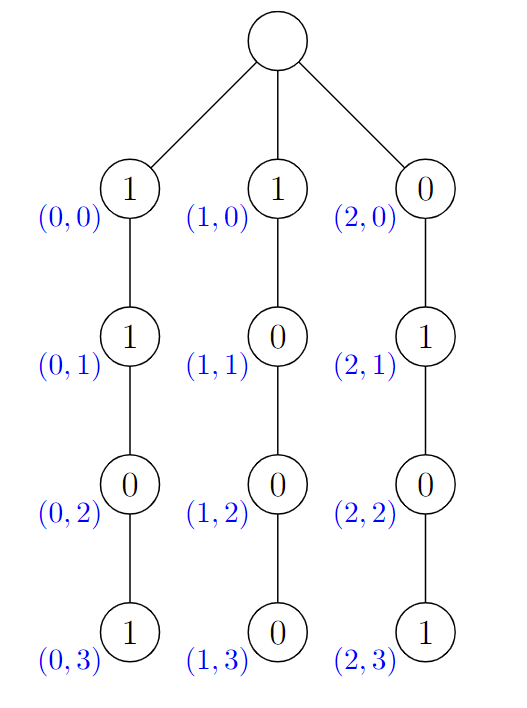}
    \caption{A binary-labeled $(12,4)$ spider, with indexing system drawn in blue to the bottom left of each vertex.}
    \label{fig:spider}
\end{figure}

An example of a binary-labeled $(12,4)$-spider is shown in \cref{fig:spider}. We index each node using two coordinates, where the first coordinate denotes the leg of the spider the node is on, and the second its depth down that leg. This is in contrast to the depth-first-search labelling in \cite{davies2019reconstructing}. In general, we may denote by $a_{i,j}\in\{0,1\}$ the label of the $(i,j)$-node and the set of all labels of $X$ as $a=\{a_{i,j}\}_{0\le i<\frac{n}{d}, 0\le j<d}$. For convenience, we define the set $S:=\{(i,j)\mid 0\le i<\frac{n}{d},0\le j<d\}$, so we can write the labels of $X$ as $a=\{a_{i,j}\}_{(i,j)\in S}$.

\subsection{Deletion channel for spiders}\label{deletion channel}

In the deletion channel for spiders, we start by independently selecting each non-root node for deletion with probability $q$. Note that we assume the root node is never deleted, as deleting the root node would disconnect the graph. When nodes are deleted, all nodes below it shift upward. If all the nodes in a leg are deleted, the entire leg disappears. If a leg disappears, the remaining legs retain the same left-to-right structure, but it is no longer clear from looking at a trace which leg in the trace corresponds to which leg in the seed.

\begin{remark}
For trees that are not spiders, one must be more careful with describing the deletion channel. Davies et al. \cite{davies2019reconstructing} studied two models, the Tree-Edit-Distance (TED) model and the Left-Propagation Model. However, in the case of spiders, both models equivalent to the deletion channel described above.
\end{remark}

For convenience in our analysis, after the deletion process we append nodes labeled $0$ to the end of each shortened leg until they are of length $d$ again. Also, if any complete legs were deleted, we add a leg of length $d$ with all nodes labeled $0$ to the right of the remaining legs. This pads the trace with $0$'s to form an $(n,d)$-spider. We refer to the resulting spider as a \emph{trace}. We remark that this padding process may cause two originally different traces to end up becoming identical. An example of the deletion and padding process is shown in \cref{fig:deletion}.

\begin{figure}
    \centering
    \includegraphics[width=9cm]{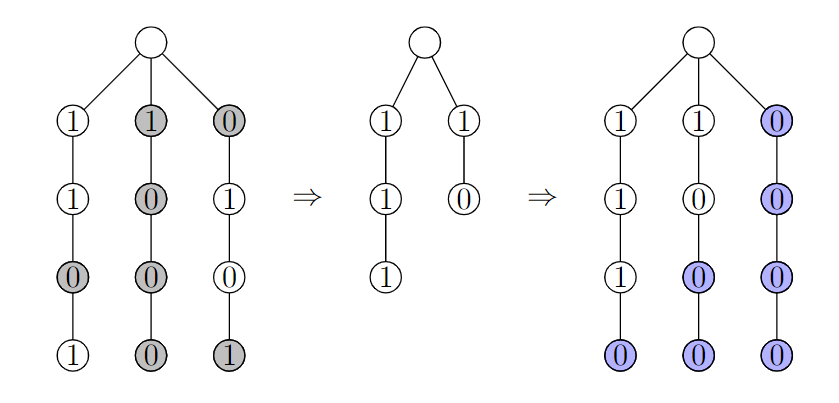}
    \caption{An example of the deletion channel applied on a $(12,4)$-spider. The deleted nodes are colored gray. Then we pad nodes labeled $0$, colored blue, to form a $(12,4)$ spider.}
    \label{fig:deletion}
\end{figure}

\subsection{Generating function for traces of spiders}

Though the deletion channel for trees is more complicated than that for strings, it turns out that one can still describe the deletion process explicitly using generating functions. These generating functions will then be used to distinguish between candidate spiders.

We begin by defining generating functions which encode the information of the possible traces of a spider:

\begin{definition}
    Let $a=\{a_{i,j}\}_{(i,j)\in S}$ denote the labels of an $(n,d)$-spider where $a_{i,j}\in\mathbb{R}$, and let the random variable $b=\{b_{i',j'}\}_{(i',j')\in S}$ denote the labels of its trace from the deletion channel with deletion probability $q$. Define a generating function $$\sum_{(i',j')\in S}b_{i',j'}w_1^{i'}w_2^{j'}$$ for each possible labeling $b=\{b_{i',j'}\}_{(i',j')\in S}$ of a trace.
\end{definition}

One of our key observations is that for each $(n,d)$-spider, we can derive a closed-form formula for the expected value of the generating function of a trace:

\begin{lemma}\label{genfunc}
Let $a=\{a_{i,j}\}_{(i,j)\in S}$ denote the labels of an $(n,d)$-spider where $a_{i,j}\in\mathbb{R}$, and let the random variable $b=\{b_{i',j'}\}_{(i',j')\in S}$ denote the labels of its trace from the deletion channel with deletion probability $q$. Define $$A_a(w_1,w_2) :=\mathbb{E}\left[\sum_{(i',j')\in S}b_{i',j'}w_1^{i'}w_2^{j'}\right]$$ to be the expected value of the generating function of a trace, where the expectation is taken over the randomness of the deletion process. Then $$A_a(w_1,w_2) = (1-q)\sum_{(i,j)\in S}a_{i,j}(q^d+(1-q^d)w_1)^i (q+(1-q)w_2)^j$$
for all $w_1,w_2\in\mathbb{C}$.
\end{lemma}

\begin{proof}
Note that the coordinates of a specific node can only decrease after the deletion process. We compute the probability that the label $b_{i',j'}$ comes from the label $a_{i,j}$, where $i\ge i'$ and $j\ge j'$. This occurs when:
\begin{itemize}
    \item $a_{i,j}$ is preserved, which occurs with probability $1-q$,
    \item Exactly $i'$ of the first $i$ paths are retained, which occurs with probability 
    \[\binom{i}{i'}(1-q^d)^{i'}q^{d(i-i')}.\]
    \item Exactly $j'$ of the first $j$ nodes in the path of the node with in $X$ with index $(i,j)$ are retained, which occurs with probability 
    \[\binom{j}{j'}(1-q)^{j'}q^{j-j'}.\]
\end{itemize}

Thus the probability that the label $b_{i',j'}$ comes from the label $a_{i,j}$ is \[(1-q) \binom{i}{i'}(1-q^d)^{i'}q^{d(i-i')} \binom{j}{j'}(1-q)^{j'}q^{j-j'}.\]
We conclude that
\begin{align*}
    \mathbb{E}&\left[\sum_{(i',j')\in S}b_{i',j'}w_1^{i'}w_2^{j'}\right]\\
    &=(1-q)\sum_{(i',j')\in S}w_1^{i'}w_2^{j'}\sum_{(i,j)\in S}a_{i,j}\binom{i}{i'}(1-q^d)^{i'}q^{d(i-i')}\binom{j}{j'}(1-q)^{j'}q^{j-j'}\\
    &=(1-q)\sum_{i=0}^{\frac{n}{d}-1}\sum_{j=0}^{d-1}a_{i,j}\sum_{i'=0}^{i}\sum_{j'=0}^j\binom{i}{i'}(1-q^d)^{i'}q^{d(i-i')}w_1^{i'}\binom{j}{j'}(1-q)^{j'}q^{j-j'}w_2^{j'}\\
    &=(1-q)\sum_{i=0}^{\frac{n}{d}-1}\sum_{j=0}^{d-1}a_{i,j}(q^d+(1-q^d)w_1)^{i}(q+(1-q)w_2)^{j}\\
    &=(1-q)\sum_{(i,j)\in S}a_{i,j}(q^d+(1-q^d)w_1)^{i}(q+(1-q)w_2)^{j},
\end{align*}
where we change the order of summation in the second equality and apply the binomial theorem in the third equality.
\end{proof}

\begin{figure}
    \centering
    \includegraphics[width=11cm]{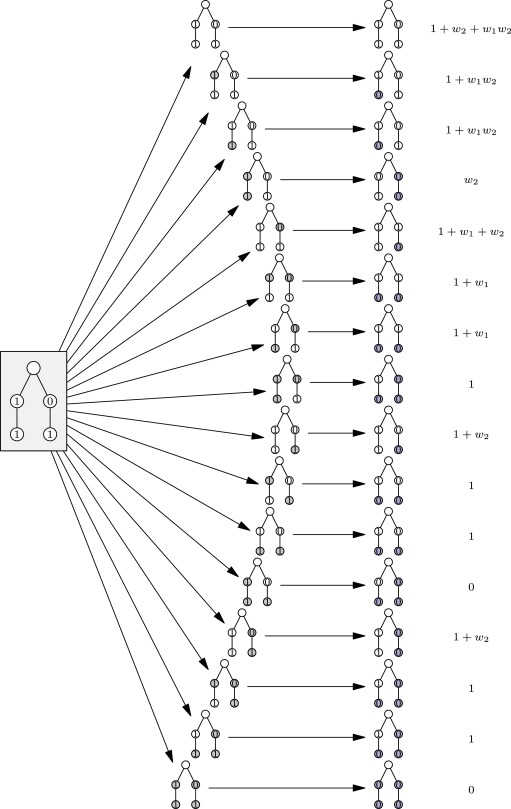}
    \vspace{-3pt}
    \caption{All possible padded traces for a certain seed $(4,2)$-spider after being passed through the deletion channel, with their associated generating functions.}
    \label{fig:expectedvalue}
\end{figure}

\begin{example}
\cref{fig:expectedvalue} depicts all $2^4 = 16$ possible deletions that could occur for a specific $(4,2)$-spider with labels $a_{0,0}=1$, $a_{1,0}=0$, $a_{0,1}=1$, and $a_{1,1}=1$, shown on the right. The figure also depicts the resulting padded traces and their associated generating functions. Note that $A_a(w_1,w_2)$, which recall is the expected value of the generating functions of the padded traces, is a weighted average of all the generating functions on the right. For example, if $q=\frac{1}{2}$, we can simply average all the values in \cref{fig:expectedvalue} to get
\begin{equation}\label{l1-example}
    \begin{split}
        A_a(w_1,w_2) &= \mathbb{E}\left[\sum_{0\le i'<2, 0\le j'<2}b_{i',j'}w_1^{i'}w_2^{j'}\right]
        \\&=\frac{13}{16}+\frac{3}{16}w_1+\frac{5}{16}w_2+\frac{3}{16}w_1w_2
    \end{split}
\end{equation}
Note that \cref{l1-example} equals what we expect from \cref{genfunc}:
\begin{align*}
    A_a(w_1,w_2) &= (1-q)\sum_{
    (i,j)\in S}a_{i,j}(q^d+(1-q^d)w_1)^i(q+(1-q)w_2)^j\\
    &=\frac{1}{2}\cdot \sum_{0\le i<2, 0\le j<2}a_{i,j}\left(\frac{1}{4}+\frac{3}{4}w_1\right)^i\left(\frac{1}{2}+\frac{1}{2}w_2\right)^j\\
    &=\frac{1}{2}\left[1+\left(\frac{1}{2}+\frac{1}{2}w_2\right)+\left(\frac{1}{4}+\frac{3}{4}w_1\right)\left(\frac{1}{2}+\frac{1}{2}w_2\right)\right]\\
    &=\frac{13}{16}+\frac{3}{16}w_1+\frac{5}{16}w_2+\frac{3}{16}w_1w_2.
\end{align*}
\end{example}

\section{Proof of main result}

In this section we prove \cref{main result}. Like in previous work on string trace reconstruction, we use a best-match algorithm to reconstruct the spider. As is typical in best-match algorithms, we compare every pair $(X^{(1)}, X^{(2)})$ of candidate spiders to see which spider from the pair is more likely to have produced the observed traces. We repeat this process for each pair of candidates and use the results to select a best possible guess for the original seed spider.

\subsection{Overview of the algorithm}

We consider all $2^n$ possible candidate spiders $X$ and select a pair of spiders to compare against each other. Suppose we select candidate spiders $X^{(1)}$ and $X^{(2)}$ with labels $a^{(1)}=\{a_{i,j}^{(1)}\}_{(i,j)\in S}$ and $a^{(2)}=\{a_{i,j}^{(2)}\}_{(i,j)\in S}$, respectively. Now, consider the element-wise difference $a=a^{(1)}-a^{(2)}$, which is nonzero since $X^{(1)}$ and $X^{(2)}$ are distinct. Let $Y^{(1)}$ and $Y^{(2)}$ denote the random traces with labels $b^{(1)}=\{b_{i',j'}^{(1)}\}_{(i',j')\in S}$ and $b^{(2)}=\{b_{i',j'}^{(2)}\}_{(i',j')\in S}$, which result from passing $X^{(1)}$ and $X^{(2)}$ respectively through the deletion channel. Now, we compute the difference of the generating functions corresponding to $X^{(1)}$ and $X^{(2)}$, which is equivalent to plugging $a$ into the expression in \cref{genfunc}:
\begin{equation}\label{diffspiders}
    \begin{split}
        &\sum_{(i',j')\in S}(\mathbb{E}[b_{i',j'}^{(1)}]-\mathbb{E}[b_{i',j'}^{(2)}])\cdot w_1^{i'}w_2^{j'}
        \\&=(1-q)\sum_{(i,j)\in S}a_{i,j}(q^d+(1-q^d)w_1)^i(q+(1-q)w_2)^j.
    \end{split}
\end{equation}
Through a process described in \cref{completing proof}, we can select some pair of indices $(I,J)\in S$, depending on $X^{(1)}$ and $X^{(2)}$, such that $|\mathbb{E}[b_{I,J}^{(1)}]-\mathbb{E}[b_{I,J}^{(2)}]|$ is lower bounded substantially. What this means is that the expected value of some label $b_{I,J}$ in the trace differs significantly depending on whether or not the seed spider was $X^{(1)}$ or $X^{(2)}$. We can use this information in combination with the empirical expected value $\mathbb{E}[b_{I,J}]$ among our observed traces to select the better match between $X^{(1)}$ and $X^{(2)}$, that is, which of $X^{(1)}$ or $X^{(2)}$ is more likely to have produced the empirical expected value $\mathbb{E}[b_{I,J}]$. Such a process is known as a \emph{mean-based algorithm}.

We repeat the above comparison for all pairs of spiders and then output the spider $X^*$ that loses against no other spiders, if such a spider exists. If no such spider exists, we can output a uniformly random spider. As the true seed spider is among the $2^n$ candidate spiders, we can use a Chernoff bound to upper bound the probability that it loses against any other candidate spider by $\mathcal{O}\left(\frac{1}{n}\right)$. Therefore, so long as we are given enough traces, the true seed spider is outputted by the algorithm with high probability.

\subsection{Littlewood polynomials}

To analyze the expression in \cref{diffspiders}, we use bivariate Littlewood polynomials from complex analysis. We begin by defining these polynomials:

\begin{definition}
    A two-variable polynomial $A(z_1,z_2)$ is called a \emph{bivariate Littlewood polynomial}
    if all of its coefficients are in the set $\{-1,0,1\}$.
\end{definition}

Note that in the right hand side of \cref{diffspiders}, the coefficients satisfy $a_{i,j}=a_{i,j}^{(1)}-a_{i,j}^{(2)}\in\{-1,0,1\}$. If we write the right hand side of \cref{diffspiders} in terms of new variables $z_1=q^d+(1-q^d)w_1$ and $z_2=q+(1-q)w_2$, then we get
\[(1-q)\sum_{(i,j)\in S}a_{i,j}z_1^iz_2^j,\]
which is $(1-q)$ times a nonzero bivariate Littlewood polynomial with $(z_1)$-degree less than $\frac{n}{d}$ and $(z_2)$-degree less than $d$. In order to lower bound this polynomial for some choice of $z_1$ and $z_2$, we prove the following lemma concerning bivariate Littlewood polynomials:

\begin{lemma}\label{littlewood}
    Let $f(z_1,z_2)$ be a nonzero bivariate Littlewood polynomial with degree $a$ in $z_1$ and degree $b$ in $z_2$. Then
    \[|f(z_1^*,z_2^*)|\ge\exp\left(-cL_1L_2\log (ab)\right)\]
    for some $z_1^*=\exp(i\theta_1)$ and $z_2^*=\exp(i\theta_2)$, where $\theta_1$ and $\theta_2$ lie in the ranges $|\theta_1|\le \frac{\pi}{L_1}$ and $|\theta_2|\le \frac{\pi}{L_2}$.
\end{lemma}
\begin{proof}
    Define the 2-variable polynomial
    \[F(z_1,z_2)=\prod_{\substack{1\le x\le L_1 \\ 1\le y\le L_2}}f\left(z_1e^{2\pi ix/L_1},z_2e^{2\pi iy/L_2}\right).\]
    Using the maximum modulus principle, which recall says that the modulus $|F|$ of any holomorphic function $F$ achieves its maximum value at the boundary of its domain, we first show that we can find some $z_1'$ and $z_2'$ on the unit circle 
    such that $|F(z_1',z_2')|\ge 1$. Note that restricting the domain of a holomorphic function to the unit disk leaves the function holomorphic.
    
    Factor $F(z_1,z_2)=z_2^k\cdot G(z_1,z_2)$ so that $G(z_1,z_2)$ no common factors with $z_2$. 
    Since $F$ has nonzero coefficients, 
    $G(z_1,0)$ can be viewed as a nonzero polynomial in one variable $z_1$. 
    We can now factor $G(z_1,0)=z_1^\ell \cdot H(z_1)$ so that $H(z_1)$ is nonzero and hence satisfies $|H(0)|= 1$. 
    By the maximum modulus principle, we can find some $z_1'$ on the unit circle such that $|H(z_1')|\ge |H(0)|= 1$. 
    We can apply the maximum modulus principle again to find some $z_2'$ on the unit circle such that $|G(z_1',z_2')|\ge |G(z_1',0)|$. 
    Therefore, we can find $z_1'$ and $z_2'$ such that
    \[|F(z_1',z_2')|=|G(z_1',z_2')|\ge |G(z_1',0)|=|H(z_1')|\ge |H(0)|= 1.\]
    Now, applying the definition of $F$ gives
    \[1\le |F(z_1',z_2')|\le |f(z_1'e^{2\pi ix/L_1},z_2'e^{2\pi iy/L_2})|\cdot (ab)^{L_1L_2-1}\]
    for all $1\le x\le L_1$ and $1\le y\le L_2$, where we use the fact that $|f(z_1,z_2)|\le ab$ for $|z_1|=|z_2|=1$. 
    We can now choose appropriate $x$ and $y$ to rotate $z_1'$ and $z_2'$ along the unit circle in the complex plane so that $z_1^*=z_1'\cdot e^{2\pi ix/L_1}=\exp(i\theta_1)$ 
    and $z_2^*=z_2'\cdot e^{2\pi iy/L_2}=\exp(i\theta_2)$
    satisfy $|\theta_1|\le\frac{\pi}{L_1}$ and $|\theta_2|\le\frac{\pi}{L_2}$. We conclude that
    \[|f(z_1^*, z_2^*)| \ge \frac{1}{(ab)^{L_1 L_2 - 1}} \ge \exp(-L_1 L_2 \log (ab)),\]
    where $z_1^*=\exp(i\theta_1)$ and $z_2^*=\exp(i\theta_2)$ satisfy $|\theta_1|\le \frac{\pi}{L_1}$ and $|\theta_2|\le \frac{\pi}{L_2}$.
\end{proof}

We remark that \cref{littlewood} is a generalization of \cite[Lemma 17]{krishnamurthy2021trace}.

\subsection{Completing the proof}\label{completing proof}
We set the parameters in \cref{littlewood} to be $L_1=L$ for some constant $L$ to be chosen later, $L_2=1$, $a \le \frac{n}{d}$, and $b \le d$. By \cref{littlewood} and the triangle inequality, we can lower bound \cref{diffspiders} as
\begin{equation}\label{2}
    \sum_{(i',j')\in S}|\mathbb{E}[b_{i',j'}^{(1)}]-\mathbb{E}[b_{i',j'}^{(2)}]||w_1^*|^{i'}|w_2^*|^{j'}\ge (1-q)\exp\left(-L\log n\right)
\end{equation}
for some $z_1^*=\exp(i\theta_1)$ and $z_2^*=\exp(i\theta_2)$ such that $|\theta_1|\le \pi/L$ and $|\theta_2|\le \pi$. Recall the change of variables
\[w_1^*=\frac{z_1^*-q^d}{1-q^d}\qquad \text{and} \qquad w_2^*=\frac{z_2^*-q}{1-q}.\]
We can upper bound $|w_1^*|$ as
\begin{align*}
|w_1^*| &=\frac{|z_1^*-q^d|}{1-q^d}
\\&\le\frac{\sqrt{\left(\cos\frac{\pi}{L}-q^d\right)^2+\left(\sin\frac{\pi}{L}\right)^2}}{1-q^d}\\
&=\frac{\sqrt{1-2q^d\cos\frac{\pi}{L}+q^{2d}}}{1-q^d}\\
&=\frac{\sqrt{(1-q^d)^2+2q^d\left(1-\cos\frac{\pi}{L}\right)}}{1-q^d}\\
&=\left(1+\frac{2q^d\left(1-\cos\frac{\pi}{L}\right)}{(1-q^d)^2}\right)^{1/2}\\
&\le\exp\left(\frac{q^d\pi^2}{2L^2(1-q^d)^2}\right),
\end{align*}
where we use the inequalities $(1+x)^r\le e^{rx}$ for $r,x\ge 0$ and $1-\cos\frac{\pi}{L}\le \frac{1}{2}\left(\frac{\pi}{L}\right)^2$. Therefore,
\[|w_1^*|^{\frac{n}{d}}\le\exp\left(\frac{n}{d}\cdot \frac{Cq^d}{L^2(1-q^d)^2}\right)\]
for some constant $C$. We can also upper bound $|w_2^*|$ as
\[|w_2^*|=\frac{|z_2^*-q|}{1-q}\le\frac{1+q}{1-q},\]
so
\[|w_2^*|^d\le \exp(C'd)\]
for some constant $C'$ depending on $q$. Therefore, by \cref{2} and the fact that $|w_1^*|, |w_2^*| \ge 1$, we have
\[\exp\left(\frac{n}{d}\cdot \frac{Cq^d}{L^2(1-q^d)^2}+C'd\right)\sum_{(i',j')\in S}|\mathbb{E}[b_{i',j'}^{(1)}]-\mathbb{E}[b_{i',j'}^{(2)}]|\ge (1-q)\exp(-L\log n).\]
Thus there exists some pair of indices $(I,J)\in S$ such that
\begin{equation}\label{eta}
    |\mathbb{E}[b_{I,J}^{(1)}]-\mathbb{E}[b_{I,J}^{(2)}]|\ge \frac{1-q}{n}\exp\left(-\frac{n}{d}\cdot \frac{Cq^d}{L^2(1-q^d)}-C'd-L\log n\right) =: \eta.
\end{equation}
Denote the right hand side of \cref{eta} by $\eta$.

Returning now to the best-match algorithm, given two candidate spiders $X^{(1)}$ and $X^{(2)}$, we define the \emph{better match} to be $X^{(1)}$ if
\[\left|\frac{1}{T}\sum_{t=1}^Ts_{I,J}^t-\mathbb{E}[b_{I,J}^{(1)}]\right|\le\left|\frac{1}{T}\sum_{t=1}^Ts_{I,J}^t-\mathbb{E}[b_{I,J}^{(2)}]\right|,\]
where $s_{I,J}^t\in \{0, 1\}$ is the value of the node at position $(I, J)$ of the $t$-th trace. Now, suppose $X^{(1)}=X^*$ is the true seed spider. For all possible seed spiders $X^{(2)}$, we can use a Chernoff bound to upper bound the failure probability, namely the probability that $X^{(2)}$ is a better match than $X^{(1)}$, by $\exp(-T\eta^2/2)$, where $T$ is the total number of traces. Therefore, by a union bound, the probability that $X^*$ loses to at least one other spider is at most
\begin{align*}
    \mathbb{P}[X^*\text{ not chosen by algorithm}]&\le \sum_{X^{(2)}\neq X^*}\mathbb{P}[X^{(2)} \text{ better match than }X^*]\\
    &\le 2^n\cdot\exp(-T\eta^2/2)\\
    &\le \exp\left(n\log 2-\frac{T\eta^2}{2}\right).
\end{align*}
For this expression to be at most $\frac{1}{n}=\exp(-\log n)$, we set
\[T = \frac{2}{\eta^2}(n\log 2+\log n) = \Theta(\eta^{-2} n).\]
Plugging in the definition of $\eta$ from \cref{eta} yields
\begin{equation}\label{T}
    T = \Theta\left(n^3\cdot \exp\left(\frac{n}{d}\cdot \frac{Cq^d}{L^2(1-q^d)}+C'd+cL\log n\right)\right).
\end{equation}
Note that the $n^3$ term is negligible. The $C'd$ term is also negligible since we are in the regime $d\le \log_{1/q}(n)$. Finally, $1-q^d\ge 1-q$ depends only on $q$, so \cref{T} can be simplified to
\[T= \exp\left(\Theta\left(\frac{nq^d}{dL^2}+L\log n\right)\right).\]
To balance these terms, we set $L=\left(\frac{nq^d}{d\log n}\right)^{1/3}$ to get a final bound of
\[T = \exp\left(C\cdot \frac{(nq^d)^{1/3}}{d^{1/3}}(\log n)^{2/3}\right),\]
where $C$ is a constant that depends only on $q$. We conclude the proof of \cref{main result}.

\section{Conclusion}

We presented a mean-based algorithm using Littlewood polynomials that reconstructs $(n,d)$-spiders with high probability in the regime $d \le \log_{1/q}(n)$, where $q$ is the deletion probability. Our algorithm uses $\exp\left(\mathcal{O}\left( \frac{(nq^d)^{1/3}}{d^{1/3}}(\log n)^{2/3}\right)\right)$ traces and works for the full range $q\in (0,1)$ of deletion probabilities.

In light of recent work improving the string trace reconstruction upper bound to $\exp(\tilde{\mathcal{O}}(n^{1/5}))$ using a non-mean-based algorithm \cite{chase2020new}, it would be interesting to see whether a similar technique could achieve an upper bound of the form $\exp\left(\tilde{\mathcal{O}}((nq^d)^{1/5})\right)$ for the spider trace reconstruction problem.

% We remark that when we set $d=c\log_{1/q}(n)$ for $c<1$,  we have $q^d=n^{-c}$ so our bound is $\exp\left(\tilde{\mathcal{O}}(n^{(1-c)/3})\right)$, which matches the result in \cite{davies2019reconstructing} but for arbitrary deletion probabilities $q<1$ (not just $q<0.7$).

\vfill

\pagebreak

\pagestyle{plain}
\bibliographystyle{plain}
\bibliography{paper}

\end{document}